\documentclass[11pt, twoside,letter]{article}

\usepackage[english]{babel}
\usepackage[utf8]{inputenc}
\usepackage{amsmath,amssymb,amsthm,graphicx,mathtools}
\usepackage[margin=2cm]{geometry}
\usepackage{hyperref}
\usepackage{xcolor,blindtext}
\usepackage{todonotes}
\usepackage{tikz}
\usepackage{pgfplots}

\theoremstyle{plain}
\newtheorem{theorem}{Theorem}[section]
\newtheorem{lemma}[theorem]{Lemma}
\newtheorem{corollary}[theorem]{Corollary}

\theoremstyle{definition}

\theoremstyle{remark}
\newtheorem{remark}{Remark}

\author{Marc Jorba-Cusc\'o$^{(1)}$\thanks{Current address: Departament de Matem\`atiques, Universitat Polit\`ecnica de Catalunya (UPC), {\tt marc.jorba@upc.edu} and {\tt daniel.perez.palau@upc.edu}}\and  Ruth I. Oliva-Z\'uniga$^{(2)}$ \and Josep Sardany\'es$^{(1)}$ \and   Daniel P\'erez-Palau$^{(3)*}$}

\title{Optimal dispersal and diffusion-enhanced robustness in two-patch \\ metapopulations: origin's saddle-source nature matters}

\date{August 2, 2023}

\begin{document}

\maketitle
{\small
\begin{itemize}
    \item[(1)] Centre de Recerca Matemàtica. Edifici C, Campus de Bellaterra, 08193 Cerdanyola del Vall\`es, Spain.  {\tt mjorba@crm.cat}, {\tt jsardanyes@crm.cat}
    \item[(2)] Universidad Nacional Autónoma de Honduras en el Valle de Sula. UNAH-VS, Boulevard UNAH-VS 21102 San Pedro Sula, Honduras. {\tt ruth.oliva@unah.edu.hn}
    \item[(3)] Escuela Superior de Ingeniería y Tecnología, Universidad Internacional de la Rioja, Av. La Paz 137, 26006 Logroño, Spain.
\end{itemize}
}
\begin{abstract}

A model of two-patch logistic metapopulation is investigated both analytically and numerically focusing on the impact of dispersal on population dynamics. First, 
the dependence of the global dynamics on the stability type of the full extinction equilibrium point is tackled. 
Then, the total population behaviour with respect to the dispersal is studied 
analytically. Our findings demonstrate that diffusion plays a crucial role in the preservation of
both subpopulations and the full metapopulation under the presence of stochastic perturbations. 
At low diffusion, the origin is a repulsor, causing the orbits to flow nearly parallel to the axes, risking stochastic extinctions. Higher diffusion turns the repeller into a saddle point. Orbits then quickly converge to the saddle's unstable manifold, reducing extinction chances. This change in the vector field enhances metapopulation robustness.
On the other hand, the well known fact that asymmetric conditions on the patches is beneficial 
for the total population is further investigated. This phenomenon has been studied in previous works 
for large enough or small enough values of the dispersal. In this work we complete the theory 
for all values of the dispersal. In particular, we derive analytically a formula for the optimal 
value of the dispersal that maximizes the total population. 

\end{abstract}


\noindent {\bf Keywords:} Dynamical systems; Bifurcations, Metapopulations, Theoretical Ecology, Stochastic extinctions


\pagestyle{myheadings}
\markboth{Optimal dispersal and diffusion-enhanced robustness in two-patch metapopulations (draft)}{Jorba-Cuscó, Oliva-Zúniga, Sardanyés, Pérez-Palau}

\section{Introduction}
The study of metapopulation theory has provided key results into the dynamics of species inhabiting fragmented populations (patches). The seminal seminal work by Levins~\cite{Levins1969,Levins1970} introduced the main models. Since then metapopulation models have been widely used to investigate the dynamics and persistence of fragmented populations under different scenarios~\cite{Clovert2012,Abbott2011}. The stability of metapopulations relies on the dynamics of the constituent subpopulations and their synchrony, where dispersal plays a significant role by affecting both subpopulation dynamics and such synchrony~\cite{Abbott2011}. Experimental evidence has confirmed the role of dispersal in mediating metapopulation stability~\cite{Ellner2001,Bonsall2002,Dey2006,Dey2012}. While some studies specifically examined the effects of dispersal rates on the dynamics and synchrony in single-species systems~\cite{Dey2006,Dey2012}, others focused on local extinctions with varying degrees of linkage between patches hosting multiple interacting species~\cite{Fahrig1985,Ellner2001,Bonsall2002,Ruiz-Herrera2018}.

Metapopulations are common and widespsread in both terrestrial and marine ecosystems, especially for species relying on dispersal to maintain interconnected populations. In terrestrial systems, some plant species in grasslands rely on seed dispersal to colonize new patches, and their populations go through cycles of local extinctions and recolonizations~\cite{Honnay2002}. The bog copper butterfly (\emph{Lycaena epixanthe}) inhabits wetlands and bogs, which are often isolated from one another, undergoing local extinctions and recolonizations within these isolated habitats~\cite{Hanski2003}. The Glanville fritillary butterfly (\emph{Melitaea cinxia}) inhabits a mosaic of meadows and patches of suitable habitat in Europe, and its populations experience regular extinctions and recolonizations~\cite{Hanski1994}. The spotted owl (\emph{Strix occidentalis}) in North America exhibits a metapopulation structure across its range, as it relies on different forest patches for nesting and foraging habitats. These forest patches are often separated by unsuitable habitats, leading to a fragmented distribution of the populations~\cite{Franklin2000}. More recently, the transient dynamics of a metapopulation of the colonial coastal bird Audouin's gull have been studied under the framework of nonlinear collective dispersal responses to biotic perturbations~\cite{Oro2023}. Metapopulations are also common in marine species such as fishes in estuaries and both rocky and coral reefs, seagrass, intertidal invertebrates, and coastal decapodes, among others~\cite{Marine_metapops2006}.

These previous examples, among many others, suggest that metapopulation theory can play a crucial role in delving into the dynamics of spatially-distributed populations having applications for conservation. 
Metapopulation dynamics have been studied with discrete-~\cite{Allen1993,Dey2014} and continuous-time~\cite{Levin74,Pulliam1998,Sardanyes2010,Sardanyes2019} models. For instance, the study of two Ricker maps with symmetric and asymmetric coupling showed that dispersal rates stabilized chaotic behaviour~\cite{Dey2014}.  In a similar direction, Allen and co-workers identified that chaotic dynamics provided robustness under local and global perturbations in coupled subpopulations modeled with logistic and Ricker maps~\cite{Allen1993}. Another model considering two-patch discrete models using coupled logistic maps studied the sensitive dependence on initial conditions for the basin of attraction of the periodic orbits, showing that chaos in one patch can be stabilized by dispersal from the other patch~\cite{Hastings1993}. Two-patch time-continuous models have been also thoroughly investigated (see e.g.,~\cite{Fang2020} and references therein). For instance, using two coupled logistic systems to study the total population for arbitrarily large dispersal rates and the impact of key parameters such as intrinsic growth rates and carrying capacities~\cite{Arditi2015}. The same system of coupled logistic models was later inspected by using the so-called balanced dispersal model instead of linear diffusion~\cite{Arditi2016}. The exploration of two-patch models with a generic growth function indicated some conditions of optimality also showing that the total population can be higher than the addition of the carrying capacities of each independent patch~\cite{Holt1985}. Two-patch models have been also investigated considering local autocatalytic growth~\cite{Sardanyes2010}.

Despite the intensive research on two-patch metapopulation models, such works often explored only a limited range of dispersal rates and focused on local dynamics. This approach hindered a comprehensive examination of potential interactions between dispersal rates and asymmetry reflected in different local dynamics within the subpopulations. For instance, some studies focused on homogeneous patches with symmetric dispersal \cite{GONZALEZANDUJAR1993,Gyllenberg1993NOV,Hastings1993,LLOYD1995}, while others allowed variations in the parameter determining dynamics between the subpopulations but maintained symmetric dispersal \cite{kendall1998}.
Conversely, certain studies examined asymmetric dispersal but limited their analysis to cases with identical qualitative dynamics in both subpopulations \cite{Doebeli1995,Ylikarjula2000}.  Other authors studied the interaction between several species in a rock-paper-scissor interaction \cite{Park2022, Wang2011}.
To gain a deeper understanding in metapopulation dynamics, it is necessary to investigate a broader range of local dynamics and dispersal rates in a systematic manner, considering their potential interactions with spatial heterogeneity and asymmetry in dispersal.
Moreover, most of these models lack results on global dynamics and have not considered stochastic perturbations, either intrinsic or extrinsic, in the overall metapopulation dynamics. As far as we know, few works have explored the interplay between dispersal, noise, and metapopulations' persistence, mainly providing numerical results in discrete-time models~\cite{Allen1993,Sardanyes2019}.

 This paper is organized as follows. In Section~\ref{sec:mamoda}, we present the model and its adimensionalization of units, which allows to reduce the system from five to three parameters without loss of generality. In subsequent subsections, we discuss the existence of equilibrium points and their stability, bifurcations and global dynamics. Most of the results regarding local dynamics were already known by the community (this is pointed out along the text). The main take-home message of these sections is that global dynamics are affected remarkably by the stability of the origin. We prove that a certain region of the phase space is foliated by 
heteroclinic connections between the global extinction and the coexistence equilibria if 
the origin is a source, but, if the origin is a saddle, there is a single heteroclinic 
connection connecting both equilibria. This phenomenology is further explored in Section~\ref{ssec:rhc}.
 In Section~\ref{sec:kicks} we provide numerical evidence on how the results identified in Sections~\ref{ssec:lgd} and~\ref{ssec:rhc} impact on the persistence of the metapopulation under stochastic population fluctuations. 
In Section~\ref{sec:odr} we conduct an analytical 
study of the dependence of the total population with respect to the dispersal. This problem 
has been addressed previously in the literature in the limit case when the dispersal 
tends to infinity (see \cite{Holt1985, Arditi2015} and references therein), and for 
a small values of the dispersal~\cite{Ruiz-Herrera2018}. In the infinite case, the authors showed that having a certain asymmetric
conditions on the patches lead to the total population to overcome the sum of the carrying 
capacities of both patches. In \cite{Ruiz-Herrera2018}, a different and less restrictive 
conditions were shown to be beneficial for small enough values of the dispersal. We here analyze
the general case (any value of the dispersal) integrating the previously studied hypotheses
in a unified theory. Moreover, we derive an analytic formula to compute the optimal value 
of the dispersal that maximizes the total population. Finally, Section~\ref{sec:con} is devoted to main conclusions.

\section{Mathematical model and dynamical aspects}\label{sec:mamoda}
We introduce the two-patch metapopulation model given by two logistic models coupled by  linear diffusion as the simplest way to model within-patch population dynamics and dispersal of individuals between patches. The model is given by the next couple of autonomous ordinary differential equations (ODEs):
\begin{equation}\label{eq:sistema_2D_previ}
\dot{x}_i = r_i \, x_i\left(1-\dfrac{x_i}{k_i}\right)+ D\cdot (x_j - x_i), \phantom{m} {\rm with} \phantom{m} i, j = 1,2;  \phantom{m} i\neq j.
\end{equation}

State variables $x_i$ represent the population numbers at patch $i$, $r_i$
the within-patch intrinsic growth rate, $k_i$ being the the carrying capacity at patch $i$. Parameter $D$ 
denotes the diffusion (dispersal) among patches, which is assumed to be symmetric and follow Fick's law. Let us consider the unit of the population as the carrying capacity of the first patch (dividing all the population variables by $k_1$) and the unit of time such that the rate of birth in the first patch is equal to 1 (dividing $r_i$ by $r_1$). This fixes $r_1=k_1=1$ without loss of generality considering dimensionless variables. We let the parameters $r_2$, $k_2$ and $D$ free for the present study. The role of the asymmetric configurations can be studied by fixing those free parameters bigger or smaller than 1.

In order to avoid cumbersome notation, we rename $r_2$ as $r$ and $k_2$ as $k$. With these modifications, the system reads as
\begin{equation}\label{eq:sistema_2D}
\begin{array}{rcl}
\dot{x}_1 &=&x_1\left(1-x_1\right)+ D\cdot (x_2 - x_1), \\[10pt]
\dot{x}_2 &=& r x_2  \left(1-\dfrac{x_2}{k}\right)+ D\cdot (x_1 - x_2).
\end{array}
\end{equation}
Notice that the parameters $k$ and $r$ are non-dimensional. As we have mentioned in the Introduction, this model has been largely studied. The motivation of this work is to obtain an analytical estimate of the dispersal rate optimizing the size of the population at the subpopulation level thus maximizing global density. Also, we will provide a detailed investigation of the dynamics close to the origin (involving full extinction) considering random fluctuations and seek how dispersal determines the fate of the population with low initial conditions in both patches and in the metapopulation under random fluctuations.

Let us now investigate the dynamics of Eqs.~\eqref{eq:sistema_2D}, focusing on the  global aspects of the phase space, namely, how trajectories of the system connect the origin and the coexistence equilibrium point. These analyses provide an important theoretical support of Section~\ref{sec:kicks}. Some of the results presented here are known but they are included here for completeness. For example, the existence of a coexistence equilibrium point was given in \cite{Freedman1977} for small values of the parameter $D$. The existence and local stability for any dispersal rate has been proved in 
several papers, as well as the global stability of the coexistence equilibrium~\cite{DeAngelis1979,Holt1985,Arditi2015}. 

\subsection{Local and global dynamics}\label{ssec:lgd}
We first inspect the number of equilibrium points for Eqs.~\eqref{eq:sistema_2D}. Since the vector field consists of two quadratic polynomials with degree two and no 
independent term, it is trivial to prove that the origin is always an equilibrium point and, by Bezout's
theorem, at most, there exist four equilibria. We now prove this elementary
result since it will be useful for further analysis that we will develop below.
\begin{lemma}[Number of equilibrium points]\label{lem:nueq}
    System~\eqref{eq:sistema_2D} has two, three or four equilibrium points.
\end{lemma}
\begin{proof}
    Let us assume, first, that $D\neq 0$ (the case $D=0$ can be studied 
    separately). The nonlinear system of equations verified by the 
    equilibrium points is: 
    \begin{align*}
        x_{1}-x_{1}^{2}+D\cdot (x_{2}-x_{1})&=0,\\
        rx_{2}-\alpha x_{2}^{2}+D\cdot (x_{1}-x_{2})&=0.
    \end{align*}
    Here, $\alpha=r/k$.
    We isolate $x_2$ with respect to $x_1$ using the first equation: 
    \begin{equation}\label{eq:ailla}
    x_{2}=\frac{1}{D} x_{1}^{2} + \frac{D-1}{D}x_{1}.
    \end{equation}
    Now, we substitute $x_2$ in the second equation, obtaining: 
    \begin{equation}\label{eq:quartic}
    -\frac{\alpha}{D^2}x_{1}^{4}-\frac{2 \alpha (D-1)}{D^{2}}x_{1}^{3}+\left ( \frac{r}{D} - \frac{\alpha (D-1)^2}{D^2} -1\right ) x_{1}^{2} + \left ( 
    1 + r -\frac{r}{D}\right ) x_{1}=0.
    \end{equation}
    
    Obviously, $x_1 = 0$ is always a solution of this quartic. Substituting $x_1=0$ in~\eqref{eq:ailla}, we get $x_2 = 0$ and we recover the equilibrium point corresponding to the extinction 
    in both patches (the origin). We can factor the later equation as $x_{1} p(x_{1})=0$, where 
    $p$ is a cubic polynomial (which always has a real solution). Two additional 
    real solutions may appear depending on the parameters $r$, $k$, and $D$. 

    The system for $D=0$ has four equilibrium points given by $(0,0)$, $(1,0)$, $(0,k)$ and $(1,k)$. 
\end{proof} 
As we said, the existence of at least two equilibrium points was given in
\cite{Freedman1977,DeAngelis1979,Holt1985,Arditi2015}. One of those points is the origin and the other one is located at the interior of the first quadrant.
The typical argument to prove the existence of the later is to show that the two nullclines
of the system cross in the interior of the first quadrant (we reproduce the 
argument in Lemma~\ref{lem:RI}).

The behaviour of the coexistence equilibrium point when $D\to \infty$ is a recurrent 
issue being thoroughly discussed in references~\cite{Freedman1977,DeAngelis1979,Holt1985,Arditi2015}. These authors 
were interested in the limit case as a model for perfectly mixing conditions and, 
more precisely, in the sum of the coordinates of the coexistence point (namely, the 
numbers at which the metapopulation is stabilized). This aspect of system~\eqref{eq:sistema_2D} will be further explored in Section~\ref{sec:odr}. For the moment being, we let the following 
remark, stating that the limit of the total population can be obtained from the quartic 
expression given by~\eqref{eq:quartic} recovering the results of previous papers in an alternative way. Before analyzing the limit cases, numerical results displaying the interior equilibrium points are shown in Fig.~\ref{fig:2D_SO} together with some orbits in phase portraits. These orbits actually provide information on how they move away from the origin and from the axes. This behaviour, as we discuss below, plays a crucial role in the fate of the species and the entire metapopulation under stochastic perturbations (see Section~\ref{sec:kicks}). The equilibrium values of the species at each patch are also displayed for different values of $D$ and $r$ in Fig.~\ref{fig:pps}.

\begin{remark}[Limit case]\label{rmk:lc}
    If $D \to \infty$, the quartic~\eqref{eq:quartic} converges uniformly to the 
    polynomial 
    $$-(\alpha+1)x_{1}^{2} + (1+r)x_{1}.$$
    Its non-trivial root is 
    $$ x_1 = \frac{1+r}{\alpha+1} =\frac{k+kr}{r+k}.$$
    On the other hand, Eq.~\eqref{eq:ailla} implies that 
    $$\lim_{D \to \infty } x_{2} (D)= \lim_{D\to \infty} x_{1}(D).$$
    It holds that, 
    \begin{equation}\label{eq:lc}
    \lim_{D\to \infty} \left ( x_{1}(D) + x_{2}(D) \right)= 2 \frac{k+kr}{r +k}.
    \end{equation}
\end{remark}

 \begin{figure}[t!]
     \centering
     \includegraphics[width=0.7\textwidth]{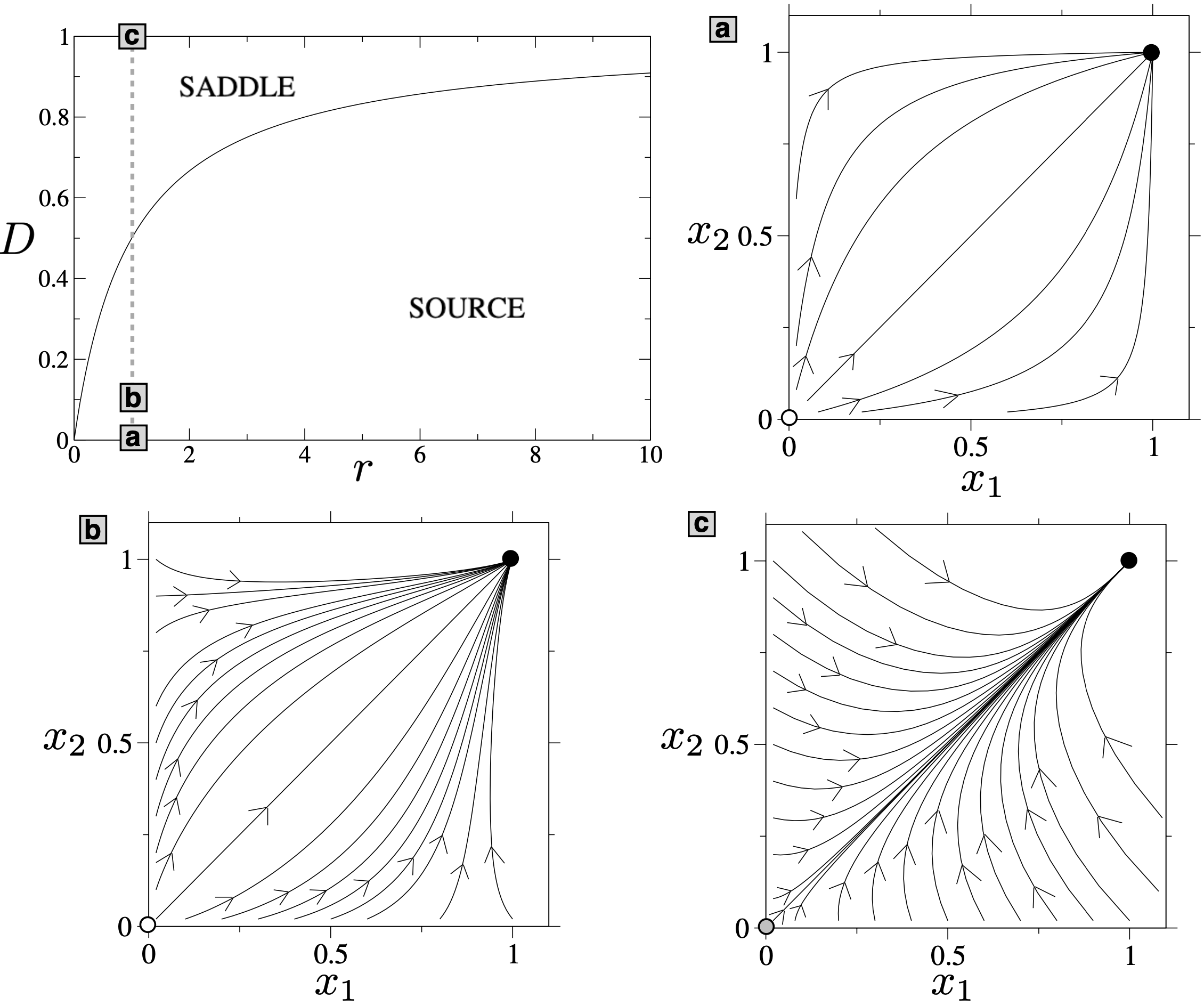}
     \caption{{\small Stability of the origin for Eqs.~\eqref{eq:sistema_2D} in the plane $(r, D)$, where the curve $D = r/(1+r)$ separates the source from the saddle behaviour. Several phase portraits are shown at increasing values of diffusion for $r = k = 1$ and: (a) $D=0$, (b) $D=0.1$, and (c) $D=1$. The arrows indicate the direction of the orbits. Stable equilibrium points are shown with black circles, while the saddle and the repeller are displayed with gray and white circles, respectively.}}
     \label{fig:2D_SO}
 \end{figure}
In \cite{DeAngelis1979}, the stability of the origin is explored. There, the authors 
proved that there always exist a repelling direction and that the origin is never a sink. Thus, extinction can never be achieved under the deterministic setting modeled with Eqs.~\eqref{eq:sistema_2D}. From 
Eq.~\eqref{eq:ailla} it is clear that the multiplicity of the origin as an equilibrium 
point is the same as the multiplicity of $x_1 = 0$ as a root of the quartic~\eqref{eq:quartic}.
\begin{remark}[Multiplicity of zero]\label{rmk:mult}
If $r=D/(1-D)$, zero has multiplicity two as solution of the 
quartic Eq.~\eqref{eq:quartic}. If, besides, $k=1/r^{2}$, 
it has multiplicity three. 
\end{remark}

\begin{lemma}[Bifurcation of the origin]\label{lem:ori}
    The origin is a repelling node for $r<D/(1-D)$ and a saddle for $r>D/(1-D)$. 
\end{lemma}
\begin{proof}
    The Jacobian matrix of the vector field evaluated at the point $(0,0)$ is given by: 
    $$
    J(0,0)=\begin{pmatrix}
1-D & D \\
D & r-D
    \end{pmatrix}. 
    $$
    An elementary computation shows that the eigenvalues are: 
    $$\lambda_{\pm}=\frac{1}{2}\left( 1- 2D +r \pm \sqrt{4D^{2} + (1-r)^2 } \right) .$$
    Notice that $\lambda_{+}$ is always positive. On the other hand, $\lambda_{-}$ is 
    negative for $r>r^{*}:=D/(1-D)$ and positive for $r<r^{*}$. Therefore, the origin 
    changes from repelling node to saddle as $r$ crosses the curve $r*$ from above 
    (see Fig.~\ref{fig:2D_SO}). 
    The type of bifurcation is determined by the multiplicity of the origin as 
    equilibrium point (see Remark~\ref{rmk:mult}): If the zero has multiplicity 
    two, the bifurcation is transcritical while if it has multiplicity three, 
    it is, generically, a pitchfork. 
\end{proof}

      \begin{figure}
          \centering
          \includegraphics[width=0.75\textwidth]{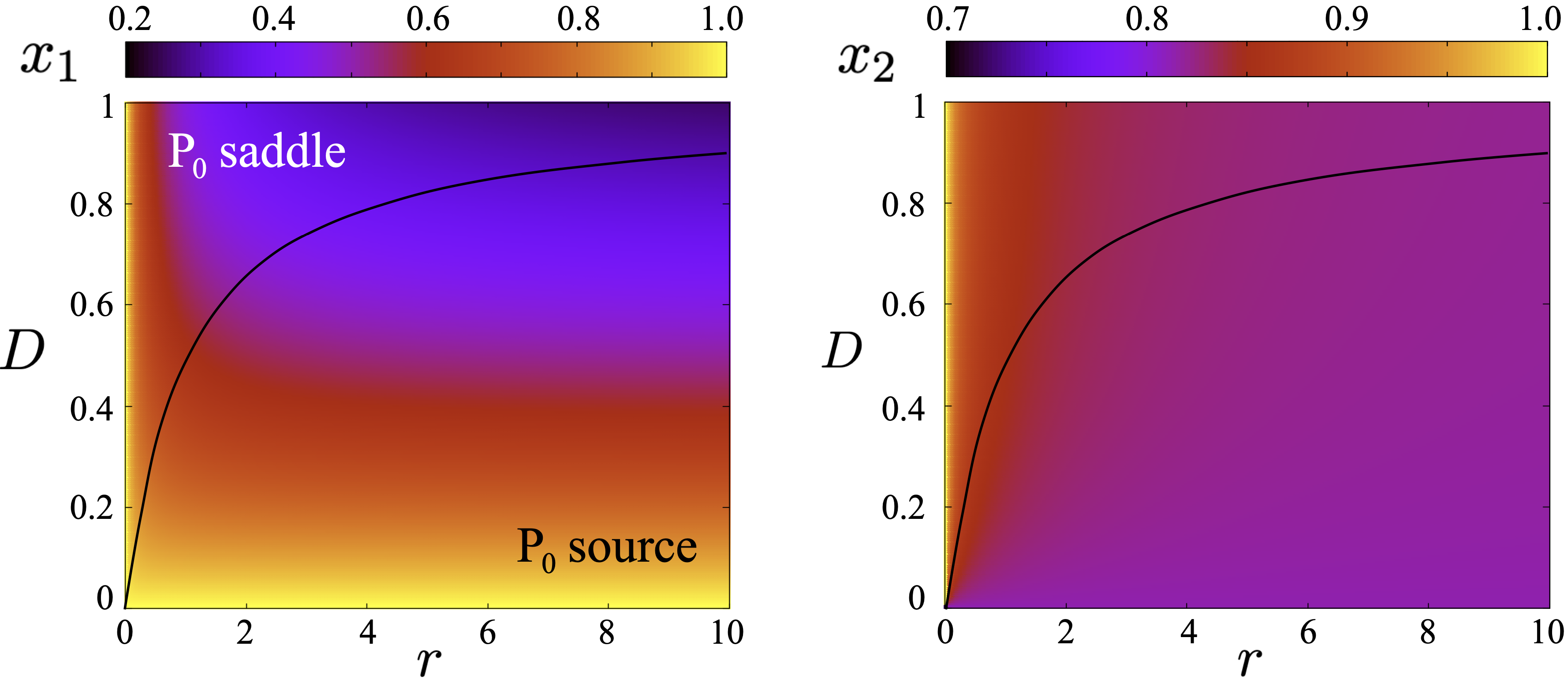}
          \caption{{\small Interior equilibrium of Eqs.~\eqref{eq:sistema_2D} in the parameter space $(r,D)$ for $x_1$ (left) and $x_2$ (right) computed numerically with $k=0.5$ and $x_1(0) = 0.45$ and $x_2(0) = 0.2$. Overlapped we display the curve separating the source from the saddle of the origin shown in Fig.~\ref{fig:2D_SO}.}}\label{fig:pps}
  \end{figure} 

The following result establishes that the nullclines enclose a region of the 
phase space which is positively invariant by the flow. This is sketched in Fig.~\ref{fig:sketch_ISO}.

 \begin{figure}[t!]
         \centering
         \includegraphics[width=0.35\textwidth]{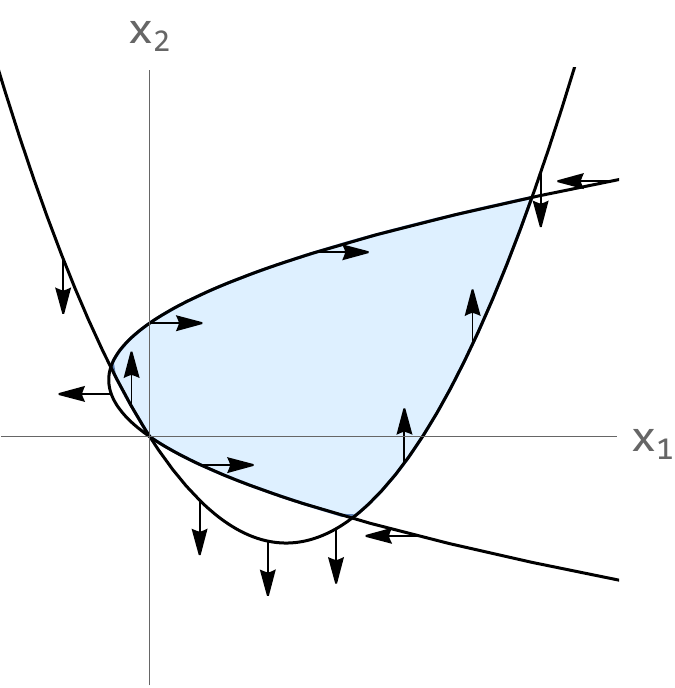}
         \hspace{0.3 cm}
         \includegraphics[width=0.35\textwidth]{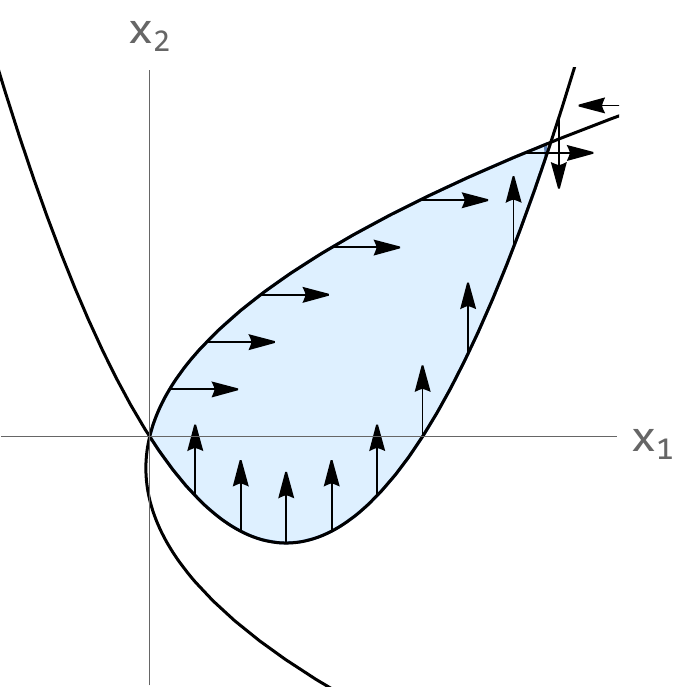}
         \caption{{\small Sketch of the positive invariant region (blue) defined by the vertical and 
         horizontal nullclines. Left: $r=1 >D$. Right: $r=0.3 <D$. In both cases  $D=0.4$, $k=0.4$. Observe that the blue region is surrounded by the isoclines and their direction points inwards.}}
         \label{fig:sketch_ISO}
 \end{figure}
 
\begin{lemma}[Intersection of the nullclines]\label{lem:RI}
For any positive values of $r$, $k$ and $D$, 
the horizontal and vertical nullclines define a positively invariant region. 
\end{lemma}
\begin{proof}
    The vertical nullcline, $\frac{d}{dt}x_{1}=0$, is described by the parabola: 
    $$x_{2}=g_{v}(x_1):=\frac{1}{D} \left( x_{1}^{2} -(1-D)x_{1} \right). $$
    The horizontal nullcline, $\frac{d}{dt}x_{2}=0$, is given by: 
    $$x_{1}=g_{h}(x_2):=\frac{1}{D}\left (\alpha x_{2}^{2} -(r-D)x_{2} \right ). $$
    Notice that the vertical nullcline crosses the horizontal axis at $N_1=((1-D), 0)$ while 
    the horizontal nullcline crosses the vertical axis at $N_2=(0, (r-D)/\alpha )$. It holds that, 
    $$\frac{d}{dx_{1}} g_{v} (N_1 ) = (D-D^2) ,$$
    and 
    $$\frac{d}{dx_{2}} g_{h} (N_2 ) = \frac{D(-D+r)}{\alpha}.$$
    
  The two parabolas intersect in a point contained in the 
        first quadrant. Indeed, if $r>D$ (see Fig.~\ref{fig:sketch_ISO} left), at the vertical axis, the horizontal 
        nullcline $g_{h}$ is above the vertical one $g_v$. Moreover, it behaves as $\sqrt{x_1}$
        when $x_1$ is large. On the other hand, the vertical nullcline 
        behaves as $x_{1}^{2}$. A direct application of the Bolzano theorem 
        shows that there exist a crossing point for $x_1 > 0$ and $x_2 > 0$.
        If $r<D$ (see Fig.~\ref{fig:sketch_ISO} right), the horizontal nullcline crosses the vertical 
        axis for $x_2=0$ and for a negative value of $x_2$. At the origin, 
        the slope of the horizontal nullcline is positive and the slope 
        of the vertical nullcline is negative. Therefore, near zero, the 
        horizontal nullcline is above the vertical one and the same argument 
        can be used to establish the existence of the crossing 
        point at the first quadrant. 
        
        From the fact that the two nullclines always cross, we get 
        a region that separates the phase space. Let us see that this region 
        is positively invariant. If $r>D$ the first component of the vector field is positive along 
        the vertical axis for values of $0< x_2 < (r-D)/\alpha$ while 
        the second component of the vector field is positive along the horizontal axis for $x_1 < (1-D)$. This, together with the fact that the parabolas 
        are nullclines establishes the positive invariance of the region. 
        If $r<D$ the result follows from the second component of the vector field 
        being positive along the horizontal axis for $x_1 < (1-D)$.
       
\end{proof}



\begin{remark}[Positively invariant region]\label{rmk:pir}
    The first quadrant is also positively invariant. The region enclosed 
    by the nullclines is not contained in it (see Fig.~\ref{fig:sketch_ISO}). For 
    the next result we shall focus on the intersection between the region enclosed 
    by the isoclines and the first quadrant. This intersection, let us
    name it $\mathcal{I}^{+}$, is also positively invariant. 
\end{remark}

\begin{remark}[Unstable eigenvector of the origin]\label{rmk:vepini}
The eigenvector associated to the $\lambda_+$ eigenvalue always points towards $\mathcal{I}^{+}$. Indeed, the eigenvector associated to $\lambda_+$ is given by:
$$
v_+=\left( \dfrac{1 - r +\sqrt{(1-r)^2 + 4 D^2}}{2 D}, 1 \right)
$$
Observe that the first component of $v_+$ is always positive, therefore, pointing to the first quadrant. If $r>D$, the isoclines crosses the axis for positive values. Therefore, always point to $\mathcal I^+$. If $r<D$, then the slope of the eigenvector must be between the slopes of $g_v$ and $g_h$:


$$\frac{D-1}{D}<\dfrac{2 D}{1 - r +\sqrt{(1-r)^2 + 4 D^2}}<\frac{D}{D-r}.$$
Both inequalities are true while $$r<D.$$
\end{remark}

\begin{theorem}[Global dynamics]\label{thm:gd}
    The following sentences hold for system~\eqref{eq:sistema_2D}:
    \begin{enumerate}
    \item There are no periodic orbits. 
    \item It has an equilibrium point ($P_3$) inside the first quadrant that is Globally Asymptotically Stable (GAS) in the first quadrant.
    \item If $D<r/(1+r)$, there exist infinitely many heteroclinic connections between the origin ($P_0$) and $P_3$.
        \item If $D>r/(1+r)$ there is a unique heteroclinic connection between $P_0$ and $P_3$.
    
    \end{enumerate}
\end{theorem}
\begin{proof}
We shall prove the statements by order.
    \begin{enumerate}
        \item  We have shown that the two nullclines cross at a point in the first quadrant.
         A periodic orbit in the first quadrant should encircle the crossing points of the 
         vertical and horizontal nullclines. But this would contradict the fact that they are 
         positively invariant. 
        \item We discuss first the case $D=0$. It is elementary to show that it has four 
        equilibrium points, $(0,0)$ being a repelling node, $(1, 0)$, $(0,k)$ being saddle points; 
        and $(1,k)$ being an attracting node. The lines $x_1 = 0$ and $x_1 =1 $ are vertical nullclines 
        and, similarly, $x_2 =0$ and $x_2=1$ are horizontal nullclines. Therefore, the rectangle 
        defined by the four equilibrium points is positively invariant. This means that there are 
        no periodic orbits around non of the equilibria and, therefore, the $\omega$-limit for a 
        full measure set of initial conditions is $(1,k)$. 
        For the case $D\neq 0$ we argue that, because of the positively invariant 
        region defined by the nullclines, the coexistence point has to be an attractor. Finally, 
        since there are no periodic orbits (in the first quadrant), it is a global attractor.
        \item If $P_0$ is a source, any initial condition on $\mathcal{I}^{+}$ with $\alpha$-limit
        $P_0$ must have $P_3$, which is GAS, as $\omega$-limit.
        \item If $P_0$ is a saddle, the unstable manifold has $P_3$ as $\omega$-limit. Any other initial condition in $\mathcal{I}^{+}$ has its $\alpha$-limit outside the first quadrant.
        
    \end{enumerate}
\end{proof}

\subsection{On the role of the heteroclinic connection}\label{ssec:rhc}
In this section we study the fact that the existence of a unique heteroclinic connection effects the global dynamics. This is done by studding the attracting character of the connection.

Lemma~\ref{lem:ori} establishes that the stability type of the origin changes from a repelling node (a source) to a saddle point (as $D$ crosses a certain critical value depending on $r$). This phenomenon has a significant impact on the dynamics near the origin, and as we will see below, will play a central role in the robustness of the metapopulation under stochastic fluctuations. Indeed, 
by Hartman's theorem, the linearized system  provides an approximation on the trajectories 
nearby. When the origin is a source, every initial condition in the positively invariant region lies on a heteroclinic connection between the origin and the coexistence equilibrium point. Additionally, two initial conditions that are very close to the origin tend to separate from each other exponentially over time.
Conversely, when the origin is a saddle point, there is a single locally attracting heteroclinic connection. In this case, nearby solutions to the origin converge rapidly to its unstable manifold and then follow the heteroclinic connection towards the coexistence equilibrium point.

In order to better understand this phenomenon, we restrict ourselves to the symmetric case with $r=1$ and $k=1$, 
 which albeit being the simplest case gathers all the elements playing 
a role in this scenario. 
When $r=1$, the bifurcation of the origin occurs at $D=1/2$. At this value of $r$, 
the eigenvalues of the Jacobian matrix evaluated at the origin are 
$$\lambda_{1}=1, ~~~~\lambda_{2}=1-2D.$$ 
and the corresponding eigenvectors are given by 
$$v_{1}=(1,1)^{T}, ~~~~v_{2}=(-1,1)^{T}.$$
Consider the linear change of variables given by the eigenvectors.   
$$
    \begin{pmatrix}
    y_1 \\
    y_2 
    \end{pmatrix}
    =
    \begin{pmatrix}
    1 & -1\\
    1 & 1
    \end{pmatrix}
    \begin{pmatrix}
        x_1 \\
        x_2
    \end{pmatrix}.
$$
Notice that the new variables $y_1$ and $y_2$ represent the difference of population between  
the patches and the total population respectively. Negative values of $y_1$ represent 
initial conditions with more individuals in the second patch. Conversely, positive 
values mean that the first patch is more populated. These variables where used 
in \cite{Holt1985} for different purposes. 
The change of variables casts the original vector field (with $r=1$) to:
    \begin{align}\label{eq:nf_wk}
    \begin{split}
        \dot{y}_1 &= (1-2D) y_1 + \left(\frac{1}{4k} -\frac{1}{4}\right) \left(y_{1}^{2}+ y_{2}^{2}\right) - \left(\frac{1}{2}+ \frac{1}{2k} \right) y_1 y_2,\\
        \dot{y}_2  &= y_2 - \left( \frac{1}{4k} + \frac{1}{4} \right ) \left(y_{1}^{2}+ y_{2}^{2} \right) - \left(\frac{1}{2}- \frac{1}{2k} \right) y_1 y_2. 
    \end{split}
    \end{align}
    
    \noindent This vector field will be used later to study how the heteroclinic connection depends on the normalized carrying capacity $k$. For the time being, we focus on the symmetric case. If we chose $k=1$, then, some terms vanish and the resulting ODE is 
    \begin{align}\label{eq:nf}
    \begin{split}
        \dot{y}_1 &= (1-2D) y_1 -y_1 y_2,\\
        \dot{y}_2 & = y_2 -\frac{1}{2}\left(y_{1}^{2}+ y_{2}^{2} \right).\\ 
    \end{split}
    \end{align}    
    Let us analyze this system more carefully. In first place, we notice that it has only 
    ecological meaning for trajectories within the region $y_2 > 0$ and $-y_1 \leq y_2 \leq y_1$. That is, when the total population is positive and the difference of the subpopulations 
    is less than the total one. It is evident that this region is positively 
    invariant by the flow. 
    In these coordinates, the four equilibria are given by 
    \begin{align*}
        P_0 &= (0,0), \\
        P_1 & = \left(\sqrt{1-4D^{2}}, 1-2D \right),\\
        P_2 & = \left(-\sqrt{1-4D^{2}}, 1-2D \right),\\
        P_3 &= (0,2).
    \end{align*}
    The Jacobian matrix in these coordinates reads
    $$
    \begin{pmatrix}
        (1-2D) - y_2 & - y_1 \\
        -y_1 & 1-y_2 
    \end{pmatrix}.
    $$

      \begin{figure}
    \centering
    \begin{tikzpicture}[scale=0.8]
    \draw[->][black, ultra thick] (0,-2.5) -- (0,6) node[anchor=east]{$y_2$};
    \draw[->][black, ultra thick] (-5,0) -- (5,0) node[anchor=north]{$y_1$};
    \draw[red, thick] (0,0) -- (5,5) node[anchor=south]{$y_{2}=y_{1}$};
    \draw[red, thick] (0,0) -- (-5,5) node[anchor=south]{$y_{2} = - y_{1}$};
    \draw[dashed][gray, ultra thick] (-5,2) -- (5,2) node[anchor=south]{$1-2D$};
    \filldraw [black] (0,0) circle (4pt);
    \filldraw [black] (0,-0.5) circle (0.pt) node[anchor=west]{$P_0$};

    \filldraw [black] (0,5) circle (4pt) node[anchor=west]{$P_3$};
    [dashed]
    \draw[dashed] [-{stealth[scale=2]}](0.05,0.5) -- (0.4,0.5);
    \draw[dashed] [-stealth](-0.05,0.5) -- (-0.4,0.5);
    
    \draw[dashed] [-stealth](0.05,1) -- (0.8,1);
    \draw[dashed] [-stealth](-0.05,1) -- (-0.8,1);

    \draw[dashed] [-stealth](0.05,1.5) -- (1,1.5);    
    \draw[dashed] [-stealth](-0.05,1.5) -- (-1,1.5);

    \draw[dashed] [stealth-](0.1,2.5) -- (2,2.5);
    \draw[dashed] [stealth-](-0.1,2.5) -- (-2,2.5);

    \draw[dashed] [stealth-](0.1,4) -- (2.5,4);
    \draw[dashed] [stealth-](-0.1,4) -- (-2.5,4);

    \draw[dashed] [stealth-](0.1,3.2) -- (2.3,3.2);
    \draw[dashed] [stealth-](-0.1,3.2) -- (-2.3,3.2);

    \end{tikzpicture}
    \caption{{\small Sketch of the phase space of System~\eqref{eq:nf}. The red lines 
    delimit the admissible region (i.e. points outside the admissible region correspond 
    to points of the original system with some negative coordinate). 
    The vertical axis is an heteroclinic connection between 
    $P_0$ and $P_3$ and it is repelling if $y_2 <2D-1$. The dashed arrows determine 
    this attracting character of the heteroclinic connection. 
    See Section~\ref{ssec:rhc} for more details.}}\label{fig:nf}
\end{figure}
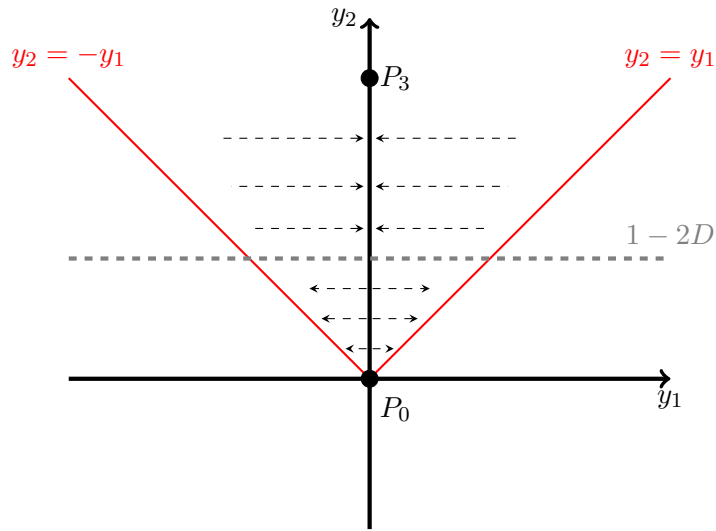

    It is easy to see that the determinant of the matrix for both $P_1$ and $P_2$ is 
    $4D^{2}-1$ which is negative for $D<1/2$. This means that these equilibria are always 
    saddle points whenever they exist. The points $P_1$ and $P_2$ lie outside the region and, 
    therefore are not ecologically meaningful. When $D=1/2$, $P_1$ and $P_2$ merge with 
    the origin in a pitchfork bifurcation. 

    The vertical axis, $\{y_{1}=0\}$, is invariant and it is, for each value of $D$, an 
    heteroclinic connection between $P_0$ and $P_3$. Moreover, it is the only heteroclinic 
    connection for $D>1/2$. Notice 
    that $\{ y_1 = 0 \}$ corresponds to the line $\{x_1 = x_2\}$ in the original coordinates. The 
    first equation of system~\eqref{eq:nf} captures the horizontal flow of the system. Let 
    us rewrite it as: 
    $$\dot{y}_1 = y_1 \left ( (1-2D)-y_2 \right).$$
    When $0<D<1/2$ and $y_2 < (1-2D)$ the horizontal flow is positive for $y_1 > 0$ and 
    negative for $y_1 < 0$, meaning that the vertical axis is repelling for $y_2 < (1-2D)$. 
    This is sketched in Fig.~\ref{fig:nf}. If $D>1/2$ the vertical axis is always attracting.

    This global property of the phase space is relevant in the capacity of the system to 
    exploit the population of one patch to recover population of the other one, at least, at 
    a short time scale. In the region of the phase space in which the vertical axis is 
    repelling, the trajectories of the system that start close to the lines $\{ y_2 =\pm y_1 \}$
    remain close to them (see Fig.~\ref{fig:nf}). This suggests that the recovering of 
    a patch is much slower when the origin is a source. A natural conjecture, to be tested in Section~\ref{sec:kicks}, is that having a saddle is favorable for recovery if one 
    of the patches is populated by few individuals, thus providing more robustness to survival to the metapopulation. The area of the sub-region in which 
    the vertical axis is attracting is given by $2D(1-2D)$ and decreases linearly with 
    $D$ until it gets zero at $D=1/2$. 

    The dependence of the unstable manifold of the origin on $k$ starts at second order. We can 
    study the second term of the Taylor expansion of the manifold by applying the parameterization method~\cite{HaroCLMF16}. That is, we 
    consider system~\eqref{eq:nf_wk} and consider a parameterization $U$ of the manifold 
    as 
    $$
    \begin{cases}
    U_{1}(s)=as^2 + \mathcal{O}(s^3),\\
    U_{2}(s)=s+bs^{2} + \mathcal{O}(s^3 ).
    \end{cases}
    $$
    Where $s$ is a real parameter and $a$ and $b$ are constants to be determined. These constants 
    can be obtained by imposing invariance of the manifold up to second order. That is, we 
    select $a$ and $b$ so the following equation is fulfilled: 
    $$Y(U(s)) = \dot{U}(s) \lambda s +\mathcal{O}(s^{3}),$$
    Here, $Y$ stands for the vector field of system~\eqref{eq:nf_wk} and $\lambda=1$ is the 
    unstable eigenvalue related to the origin. Solving the latter equation for $a$ and $b$
    results in 
    $$a=\frac{1-k}{4k(1+2D)},~~~~~b=-\frac{1+k}{4k}.$$
    Looking at the sign of $a$ we can understand how the manifold bends (at second order) when 
    $k\neq 1$. As expected, if $k<1$ the manifold bends towards the region $y_1 >0$, where 
    the population of the first patch is larger than the population of the second patch. The
    situation is opposite if $k>1$.

    We have seen that when the origin is a source the line $x_1=x_2$ is locally repealing and gets attracting as the origin becomes a saddle. To do so, we have used adapted coordinates ($y_1$-$y_2$). 
    This allows us to conjecture that having a saddle benefits the recovering capacity of the system. This fact is studied in next section introducing random population fluctuations. 


\section{Role of diffusion in metapopulaiton robustness to perturbations}\label{sec:kicks}
  
In Section~\ref{ssec:rhc}, we analyzed the local dynamics close to the origin and concluded that if $D$ is chosen such that the origin is a saddle-type equilibrium point, the heteroclinic connection is locally attracting. In this section, we provide numerical evidence showing how this can be beneficial to the persistence of the metapopualtion under the effect of stochastic perturbations. Specifically, we conjecture that a homoclinic connection that attracts initial conditions nearby may have a positive effect on the persistence of the metapopulation when one of the two subpopulations is close to extinction.

To investigate this phenomenon, we designed a simple numerical experiment simulating random deaths due to some external agent (e.g., sick animals, predators, climatic factors, hunting). The process is as follows:
\begin{enumerate}
\item We consider system~\eqref{eq:nf} and fix a maximal integration time of $T=10$ system units.
This value has been selected experimentally to ensure effective stabilization for all the initial 
conditions in the unperturbed system. Here, effective stabilization means that the initial 
condition is closer to $P_3$ than $10^{-3}$ in $\ell_{2}$-norm. Notice that, in these symmetric 
conditions, $P_3$ does not depend on $D$.
\item We select a grid of $500\times 250$ initial conditions in the region $0< y_{2} <4/5$ and $-y_2 < y_1 < y_2$. Note that $0< y_{2} < 4/5$ are the values of $y_2$ for which the vertical axis is repelling if we pick $D=0.1$.
\item For each initial condition, we perform an integration and, at random intervals following a uniform distribution $\mathcal U(0,0.1)$,
we subtract a random quantity from the total population $y_2$. This quantity also  follows a uniform distribution $\mathcal U(0,1/50)$ 
i.e., the maximum subtracted value is the total carrying capacity divided by $100$.
\item If this random perturbation is large enough to cause $y_2$ to become negative, the integration is terminated, and we consider the population to be extinguished. If $y_1 > y_2$ i.e., $x_2 < 0$ in the original coordinates, we then rearrange the initial condition to $y_1 = y_2$, and the integration continues. A similar procedure is done if $y_2 < - y_1$ i.e., $x_1 < 0$. When the integration is completed, any initial conditions that result in one of the subpopulations being smaller than $1/100$ are labeled as being at risk of extinction in patch 1 (or 2).

\end{enumerate}
We perform this experiment (steps 1-4) $N$ times for each initial condition,
associating an extinction probability with each initial condition. The extinction probability 
associated to each initial condition is computed by dividing the number of simulations that
result to extinction by the total number of simulations $N$ (we have chosen $N=100$).
  \begin{figure}
          \centering
          \includegraphics[width=\textwidth]{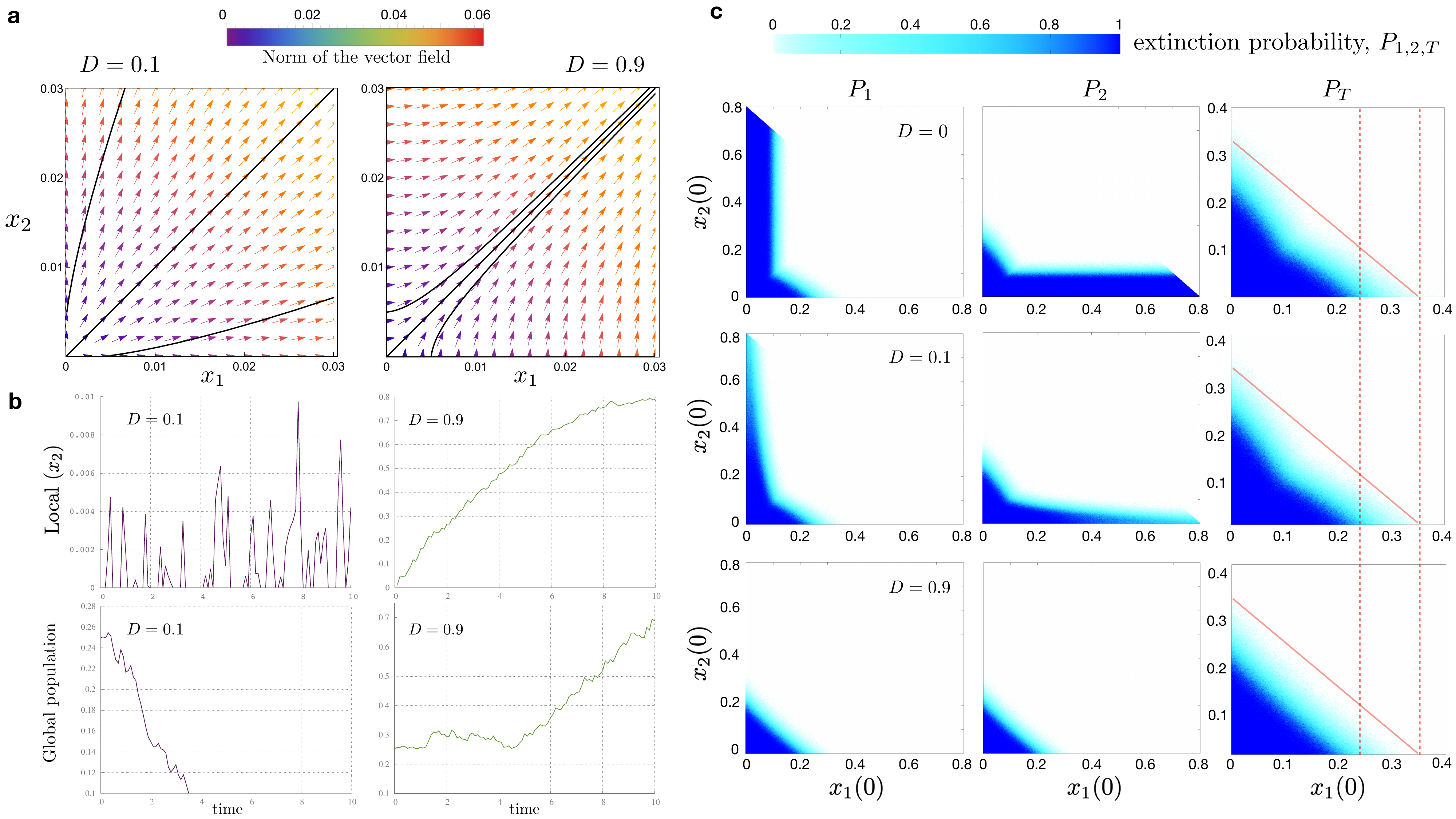}
          \caption{{\small (a) Vector field (the colors of the arrows denote the norm of the field) of the system~\eqref{eq:sistema_2D} for $D=0.1$ (left) and $D=0.9$ (right). (b) Time series with local ($x_2$) and global populations for $D=0.1$ and $D=0.9$. 
          (c) Extinction probabilities ($P_{1,2}$ and $P_T$ denote local and global extinctions, respectively) for initial conditions close to the origin for $D=0$ (first row), $D=0.1$ (second row) and $D=0.9$ (third row). The dashed red lines indicate the values of $x_1(0)$ where extinctions are observed. The red diagonal line indicates the union of these values including $x_2(0)$ to ease the visualization of the extinction regions. In all the panels we fix $k=r=1$.
          }}\label{fig:exts}
  \end{figure}

In Fig.~\ref{fig:exts}(a) we display the direction of the vector field of system~\eqref{eq:sistema_2D} for $D=0.1$ and $D=0.9$
(under symmetric conditions, $k=r=1$). In each of the pictures there appear three trajectories, 
the heteroclinic connection $\{x_1 = x_2\}$ and the trajectories starting at $(10^{-2}, 0)$ and $(0, 10^{-2})$ respectively. When the origin is a source ($D=0.1$) these trajectories spend some time 
close to the axes. When $D=0.9$, both trajectories are rapidly attracted by the heteroclinic 
connection (Fig. \ref{fig:exts}(a)). 
These two plots provide an illustration of the 
main point of Section~\ref{ssec:rhc}, namely, the system is better at recovering 
a patch which is close to extinction if the origin is a saddle.

Figure~\ref{fig:exts}(b) displays some representative time series with local and global extinctions together with survival dynamics at increasing diffusion simulating the random perturbations described above. 
The first row displays the dynamics of the population at patch $2$ ($D=0.1$ first column, $D=0.9$ second column). In the upper panel for $D=0.1$, the local population goes to extinction but recovery is still possible from the other patch although the population remains at extremely low values. The lower panel for this diffusion value displays another simulation with full extinction of the metapopulaiton. In the case $D=0.9$, the attracting 
character of the heteroclinic connection permits the total population to exit the extinction zone rapidly and, therefore, the population persists. The upper plot clearly shows this effect and the population at patch $2$ rapidly increases. The lower panel displays another simulation for the entire metapopulaion, which already has large population values at initial times while it progressively grows towards larger values. We notice here that the population continues increasing after $time = 10$ (results not shown).

Finally, panel (c) displays both local ($P_{1,2}$) and global (total, $P_T$) extinction probabilities for $D=0$ (first row), $D=0.1$ (second
row) and $D=0.9$ (third row) in the space of initial conditions $(x_1(0), x_2(0))$. The case $D=0$ is included for completeness. We notice that extinction probabilities shall be interpreted differently for local total populations. 
Total extinction involves that both time series $x_1$ and $x_2$ abandon the first quadrant simulatenously and the integration is 
terminated. The extinction probability for subpopulations refer to the event \textit{the number of individuals at patch $i=1,2$ is below $1/100$ after $10$ units of time}. In this case, the ecological interpretation is: if the subpopulation starts in the danger zone, then the  system is not able to take it out from it with a certain probability. We notice that, for  $D=0.1$ the probability of subpopulation extinction is close to one along the axis when 
the line $\{x_1 = x_2 \}$ is not attracting. In the case $D=0.9$ the system is always able  to recover a patch in danger unless the total population becomes extinct. These analyses also show that at low diffusion values, both local and global extinctions occur for wide ranges of initial conditions. A slight increase in diffusion (from $D=0$ to $D=0.1$) decreases local extinctions although global ones are still found. The simulations done with $D=0.9$ indicate that local extinctions only take place for smaller population sizes and global extinctions are also restricted to low population values. As expected, diffusion ensures persistence of the populations at a local level keeping the entire metapopulation in a safe state.

\section{Optimal dispersal rate}\label{sec:odr}
A well studied phenomenon from metapopulation theory is the fact that, under suitable conditions, 
the total population can exceed the sum of the carrying capacities of the patches~\cite{Holt1985}. This is known 
to happen when the conditions in one of the patches are better than the conditions in the other 
one. When this hypothesis is fulfilled, the less advantageous patch acts as a source of population. 

For fixed values of any pair of values $r$ and $k$, system~\eqref{eq:sistema_2D} has always 
a coexistence equilibrium point $P_{3}(D)$ that is globally asymptotically attracting. Therefore, for any selection of the parameters of the system and any initial condition, the population tends to stabilize at $P_{3}(D)$. That is, for any threshold $\varepsilon$, any values $r$, $k$ $D$, and 
any initial condition $(x_{1}(0), x_{2}(0))$, there exist a time $t$ for which 
$$\left \| \varphi^{D,r,k}_{t}(x_{1}(0), x_{2}(0)) - P_{3}(D) \right \|_{2}< \varepsilon.$$ 
Here, $\varphi^{D,r,k}_{t}$ denotes the flow of system~\eqref{eq:sistema_2D}. Therefore, the behaviour of the 
coordinates of $P_{3}(D)$ with respect to $D$ determines the long-term behaviour of the total 
population ($x_1+x_2$) of the system for any initial condition. In \cite{ruizherrera2018}, the authors 
define the function 
$$\Omega(D):= \| P_{3} (D)\|_{1}.$$
They are able to prove the following result.
\begin{lemma}[\cite{ruizherrera2018}]\label{lem:rh}
    $$\Omega'(0)=\left ( \frac{1}{r} -1\right)(1-k).$$
\end{lemma}
Henceforth, the function $\Omega$ is increasing nearby $D=0$ if $r>1$ and $k>1$ or if $r<1$ and $k<1$. 
As $\Omega(0)>1+k$, under these conditions, $\Omega$ can exceed the sum of the 
carrying capacities of the patches. The behaviour of $\Omega$ has been analyzed previously for arbitrarily 
large diffusion values (namely $D\to \infty$). For instance, results in \cite{Holt1985, Arditi2015} show 
that dispersal is beneficial if $k>1$ and $r/k > 1$ or $k<1$ and $r/k < 1$. These conditions 
are usually refereed as positive (negative) $r-k$ relation. As it is stated in \cite{ruizherrera2018}, this last condition is more restrictive than  $r>1$ and $k>1$ or if $r<1$ and $k<1$. In \cite{Arditi2015}, the authors derive the following formula: 
$$ \lim_{D\to \infty} \Omega(D)=1+k + (1-k) \frac{k-r}{k+r}, ~~~~k\geq 1, ~~r/k > 1.$$
Notice that formula~\eqref{eq:lc} from Remark~\ref{rmk:lc} is a compact version of 
this one. In particular, using~\eqref{eq:lc} it can be shown that $\Omega(\infty)>1+k$
whenever $k$ is contained between $1$ and $r$ (without assuming $r>1$ or $r<1$).

The qualitative behaviour of the function $\Omega$ with respect to $D$ is also 
analysed in \cite{Arditi2015}. It is shown that, for some values of the parameters the function 
is strictly increasing while for some other, it has a maximum. Asking which 
dispersal rate maximizes $\Omega$ is a natural question. In this section we provide 
an answer and, moreover, recover some insights on the qualitative behaviour of 
the function $\Omega$. 

The first step to tackle the problem is to understand the locus of the equilibrium points of 
system~\eqref{eq:sistema_2D}. Particularly, while the system is two dimensional, the equilibrium 
points are located, for fixed values of $k$ and $r$, in a closed 1-dimensional manifold given by an ellipse, as the following result shows.

\begin{lemma}[Equilibria lie on an ellipse]\label{lem:equieli}
    The equilibrium points of system~\eqref{eq:sistema_2D} are contained on an ellipse with 
    center $(1/2, k/2)$ and axes $\sqrt{a}$ and $\sqrt{b}$, where: 
    \begin{align*}
        a&=(1+rk)/4,\\
        b&=(k^{2} + k/r)/4.
    \end{align*}
\end{lemma}
\begin{proof}
We consider once again the system of equations for the equilibria:
\begin{align*}
    x_{1}-x_{1}^{2}+D\cdot (x_{2}-x_{1})&=0,\\
    rx_{2}-\frac{r}{k} x_{2}^{2}+D\cdot (x_{1}-x_{2})&=0.
\end{align*}
Adding both equations, we get the following expression. 
$$x_1 - x_{1}^{2} + r x_2 - \frac{r}{k} x_{2}^{2}=0,$$
which is equivalent to:  
\begin{equation}\label{eq:eli}
\frac{(x_1 -\frac{1}{2})^{2}}{a} + \frac{(x_2-\frac{k}{2})^{2}}{b}=1,
\end{equation}
This is, in fact, the 
equation of an ellipse centered at $(1/2, k/2)$ and with semi-axes $\sqrt{a}$ and $\sqrt{b}$.  
\end{proof}
The previous argument only states that, if $(x_1 , x_2)$ is an equilibrium 
point, then it is contained in the ellipse. The converse is tackled in the 
following Lemma for some points of the ellipse. 
\begin{lemma}[The points of the ellipse are equilibria]\label{lem:eliequi}
    If $(x_1 , x_2)$ verify Eq.~\eqref{eq:eli}, the following identity holds: 
    \begin{equation}\label{eq:D2BE}
        \frac{x_{1}^{2} - x_{1}}{x_{2}-x_{1}} = \frac{\alpha x_{2}^{2} - r x_{2}}{x_{1}-x_{2}}, ~~~~\alpha=\frac{r}{k}.
    \end{equation}
    Moreover, the point $(x_1 , x_ 2)$ is an equilibrium point of system~\eqref{eq:sistema_2D} for 
    $$
    D=\frac{x_{1}^{2} - x_{1}}{x_{2} -x_{1}}. 
    $$
\end{lemma}
\begin{proof}
    Notice that the second statement is trivial if identity~\eqref{eq:D2BE} holds. Before 
    proving~\eqref{eq:D2BE}, we observe that, if $a$ and $b$ are defined as in the 
    statement of Lemma~\ref{lem:equieli}, then: 
    $$
    \frac{a}{b}= \frac{1+\alpha k^2}{k^2 + \alpha^{-1}} = \alpha.
    $$
    Notice also that, by square completion, 
    $$x_{1}^{2} -x_{1}=(x_1 - 1/2)^{2} -1/4.$$
    If $(x_1 , x_2 )$ verify Eq.~\eqref{eq:eli}, then the above quantity is equal to 
    $$a - (1/4 + k^{2} \alpha/4) - (\alpha x_{2}^{2} - \alpha k x_{2}).$$
    The identity is proven by noticing that 
    $$(1/4 + k^{2} \alpha / 4) = a,$$
    and recalling that $\alpha k = r$. Indeed, 
    $$
     \frac{x_{1}^{2} - x_{1}}{x_{2}-x_{1}} = \frac{- (\alpha x_{2}^{2} - \alpha k x_{2})}{x_{2}-x_{1}}= \frac{\alpha x_{2}^{2} - r x_{2}}{x_{1}-x_{2}}.
    $$
\end{proof}
We stress the fact that identity~\eqref{eq:D2BE}  does not provide 
a one-to-one map from the ellipse to the set of equilibria of system~\eqref{eq:sistema_2D}. For 
instance, the origin always belongs to both the ellipse and the set of equilibria but identity~\ref{eq:D2BE} does not make any sense if $x_1 = x_2 = 0$. 
On the other hand, if $x_{2} = x_{1}$, the value 
of $D$ cannot be recovered from identity~\eqref{eq:D2BE}. For the coexistence equilibrium point, 
this happens, for instance, in the case $r=1$, $k=1$. 
In this case, however, the coexistence equilibrium does not depend on $D$ (see Section~\ref{ssec:rhc}).  

To obtain the best dispersal rate i.e. the one that maximizes the total population, one has 
to look for the maximum of the function $\Omega$ on the ellipse~\eqref{eq:eli}. Notice that as 
the objective function is continuous and the set of feasible solutions a compact, this problem 
has always a solution. However, the solution may correspond to a negative value of 
the dispersal rate $D$. Interestingly enough, the singularity of identity~\eqref{eq:D2BE} implies 
that, for some choices of $r$ and $k$, the population is maximized at the limit $D\to \infty$. The 
next result provides the solution of the aforementioned optimization problem and it is tackled 
using Lagrangian Multipliers.  

\begin{theorem}[Optimal dispersal rate]\label{thm:odr}
    Let $(\bar{x}_{1}(D), \bar{x}_{2}(D))$ be the coexistence solution of system~\eqref{eq:sistema_2D} and
    $\Omega(D)=\bar{x}_{1}(D)+\bar{x}_{2}(D)$, then: 
    \begin{equation}\label{eq:OTP}
    \Omega(D)\leq \frac{1}{2} \left ( \sqrt{1+ \frac{r^2 +1}{r}k + k^2 } +(1+k)\right)=:\Omega^{*}.
    \end{equation}
    Moreover, this bound is sharp: There exist $D^{*}$ for which $\Omega(D^{*})=\Omega^{*}$.
   
\end{theorem}
\begin{proof}
    In view of Lemma~\ref{lem:equieli}, the maximal total population can be 
    obtained by maximizing the function $\Omega$ on the ellipse~\eqref{eq:eli}. This 
    is equivalent to solve the following optimisation problem:
    \begin{align*}
         \text{maximize}~~~& X_1 + X_2 + \frac{1}{2}(1+k),\\
         \text{subject to}~~~& \frac{X_{1}^{2}}{a}+\frac{X_{2}^{2}}{b}=1.
    \end{align*}
    Here, $X_1 = x_1 -1/2$ and $X_2 = x_2 - k/2$. 
    To solve the problem, we name $f(X_1, X_2)=X_1+X_2 + \frac{1}{2}(1+k)$ the 
    objective function and $H(X_{1}, X_{2})=\frac{X_{1}^{2}}{a}+\frac{X_{2}^{2}}{b}-1$ the 
    restriction function. Let us maximize the Lagrangian function. 
    $$\mathcal{L}(X_1 , X_2)=f(X_1 , X_2)+ \lambda H(X_1 , X_2),$$
    for some real parameter $\lambda$. If we look for the roots of $\nabla \mathcal{L}$
    we notice that those are: 
    $$ \left (\pm \frac{a}{\sqrt{a+b}}, \pm \frac{b}{\sqrt{a+b}} \right).$$
    Since we are looking for a maximum of $f$, it is clear that we must use the solutions with 
    positive sign (the negatives, which have no ecological meaning, correspond to the global minimum). 
    This optimal point being an equilibrium point comes as a straightforward application of Lemma~\ref{lem:eliequi}. 
    The maximal value of $f$ is given by
    $$\sqrt{a+b} + \frac{1}{2}(1+k),$$
    which corresponds to $\Omega^{*}$ if we expand $a$ and $b$ in terms of $k$ and $r$. 
    Moreover, if we we recast $(X_1, X_2)$ to the original coordinates, we can use 
    Lemma~\ref{lem:eliequi} to find the value of $D$ for which the optimum is an equilibrium point:
    \begin{equation}\label{eq:ODR}
    D^{*}=\frac{x_{1}^{2} -x_{1}}{x_{2}-x_{1}},
    \end{equation}
    with 
    $$x_{1}=\frac{a}{\sqrt{a+b}} +\frac{1}{2},~~~ x_{2}=\frac{b}{\sqrt{a+b}} + \frac{k}{2}.$$
\end{proof}
\begin{figure}[t!]
\centering
\includegraphics[width=0.475\textwidth]{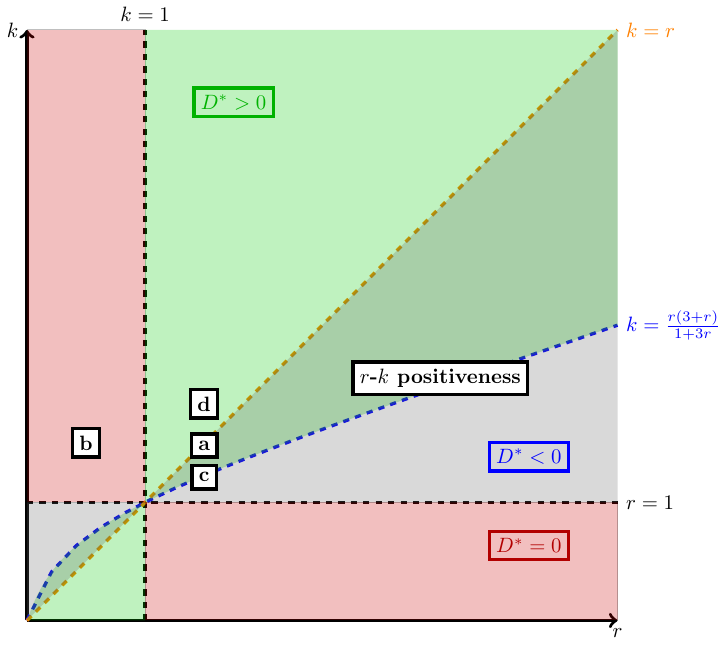}
 \caption{{\small Behaviour of the total population $\Omega$ in the parameter space ($r, k$). Green areas denote 
 those values of the parameters for which the total population is achieved for a positive 
 value of the dispersal rate. Grey ones, are areas where $r$-$k$ positiveness is fulfilled. Note
 that green and grey areas overlap. Grey area which do not overlap with green area are 
 those values for which the total population 
 increases monotonically for positive dispersal rate but the optimal value is not achieved. 
 Red regions denote those values of $r$ and $k$ for which the total population decreases 
 for positive dispersal rates. See Corollary~\ref{cor:btp} for more details. Boxes with letters from \textbf{a} to \textbf{d} correspond to the values of $r$ and $k$ used in Fig.~\ref{fig:TPvsD} where the total population is shown at increasing values of $D$.
 }}\label{fig:btp}  
 \end{figure}
 
There are several consequences of Theorem~\ref{thm:odr}. Notice first that, if $r=1$, then 
$\Omega(D)=1+k$. This means that, if the two patches have the same growth rate, the total 
population is maximized at the sum of the carrying capacities. This is a known 
fact that can be recovered from Eq.~\eqref{eq:OTP}. Moreover, the linear coefficient 
of the quadratic polynomial inside the square root is 
$$\frac{r^{2}+1}{r},$$
which is a function whose minimal value is $2$ and it is achieved precisely at $r=1$. That is,
$$\Omega^{*}(D) \geq \frac{1}{2} \left (  \sqrt{1+ 2k +k^{2}} +(1+k)\right)=(1+k).$$
Hence, $\Omega^*$ is always superior (or equal) to the sum of the carrying capacities of the patches. 
As we stated before, this does not mean that the optimal total population is achieved for positive $D$. To check when 
the optimal population is achieved for a positive value of the dispersion rate we have to 
examine the sign of Eq.~\eqref{eq:ODR}. The parameter space $(r,k)$ can be partitioned 
according to the behaviour of the total population with respect to the dispersion rate (see Fig.~\ref{fig:btp}). The 
following result provides this classification.
\begin{corollary}[Behaviour of the total population]\label{cor:btp}
    $\Omega^{*}$ is achieved for a positive dispersion rate $D^{*}$ if one 
    of the following conditions hold: 
    \begin{enumerate}
        \item $r>1$ and $k>k^{*}(r)$, 
        \item $r<1$ and $k<k^{*}(r)$,
    \end{enumerate}
    where 
    $$k^{*}(r)=\frac{r(3+r)}{1+3r}.$$
    Moreover, if $r<1$ and $k^{*}(r)<k<1$ then $\Omega'(D)>0$ for $D>0$. 
    In any other case, $\Omega(D)<1+k$ if $D>0$.
\end{corollary}
\begin{proof}
    From the proof of Theorem~\ref{thm:odr} we know that $\Omega$ has only one critical value 
    that is an equilibrium point inside the first quadrant. We also know that the maximal 
    value is achieved at: 
    $$ D^{*}=\frac{x_{1}^{2} -x_{1}}{x_{2}-x_{1}},$$
    where, 
    $$x_{1}=\frac{a}{\sqrt{a+b}} + \frac{1}{2},~~~ x_{2}=\frac{b}{\sqrt{a+b}}+\frac{k}{2}.$$
    Of course, $D^{*}$ is positive whenever the numerator and the denominator have equal sign. The numerator has positive sign if and only if $r>1$. On the other 
    hand, the denominator has positive sign if and only if 
    $k>r(3+r)/(1+3r)$.
    This proves points $1$ and $2$. If $r<1$ and $k^{*}(r)<k<1$, the optimal 
    value is achieved for some negative value of $D$. However, $\Omega(D)$ is increasing 
    near $D=0$ (see Lemma~\ref{lem:rh}). As there are no other relative optimal values of $D$, $\Omega(D)$ must be 
    an increasing function of $D$ for $D>0$.
\end{proof}

\begin{remark}
    For $k= r(3+r)/(1+3r)$, the optimal value is achieved in perfectly mixing conditions, 
    that is, $D^{*} \to \infty$. 
\end{remark}
 \begin{figure}[t!]
     \centering
      \includegraphics[width=0.65\textwidth]{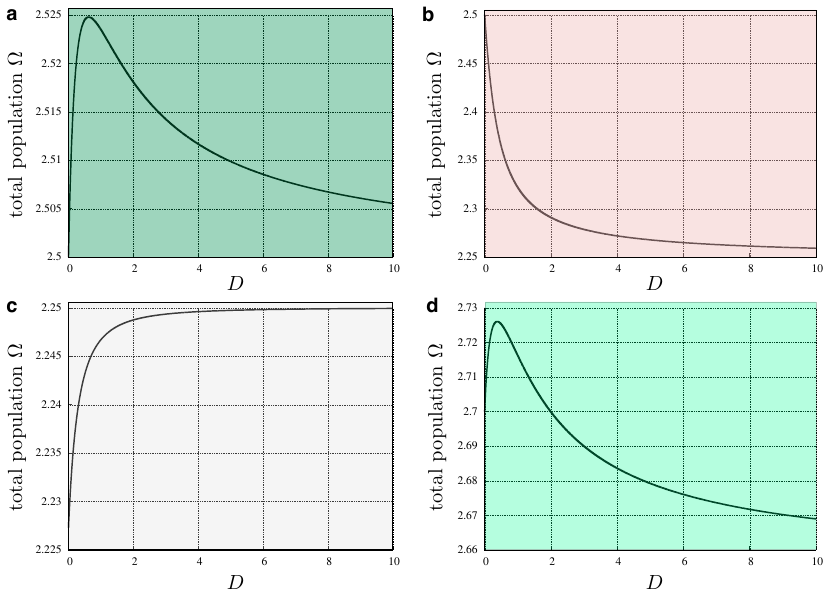}
     \caption{{\small Total population $\Omega$ with respect to diffusion with: (a)  $r=1.5$, $k=1.5$, (b) $r=0.5$, $k=1.5$, (c) $r=1.5$, $k=k^{*}(1.5)\approx 1.227$, and (d) $r=1.5$, $k=1.7$. The colors of the plots match with the colors of the previous figure.}}
     \label{fig:TPvsD}
 \end{figure}

The information provided by Corollary~\ref{cor:btp} is shown in  Fig.~\ref{fig:btp}. The 
horizontal axis represents the value of $r$ and, the vertical, the value of $k$. Different 
regions are colored according the behaviour of the function $\Omega$ with respect to the 
dispersal rate. The red zone represent values of the parameters $r$ and $k$ for which having 
positive $D$ is always detrimental in terms of total population.This is because the function 
$\Omega$ is monotonically decreasing with respect to $D$. Zones which are not coloured
in red correspond to the conditions appearing in \cite{Ruiz-Herrera2018}. Therefore, 
in non-red zones $\Omega$ is locally increasing for $D$ near to zero. 
We coloured in gray regions in which $r-k$ positiveness condition is fulfilled. In this 
region, the population at the limit case $D\to \infty$ is larger than the sum of the 
carrying capacities of the patches (see, for instance, \cite{Freedman1977}).
 The line $k=r$ determines the limit case 
$\Omega(\infty) < 1+ k$ if $k>r$ and $r>1$ or $k<r$ and $r<1$. The region determined by the 
curves $k=r(3+r)/(1+3r)$ and $k=r$ correspond to values for which $\Omega^{*} > 1+k$ but 
$\Omega(\infty) < 1 + k$. The green regions represent values 
of the parameters for which $\Omega^{*}$ is achieved for some positive $D$. Notice 
that the green region and the gray region overlap (this overlapping is seen as a dark 
green). For values of the parameters in this overlapped region, $D^*$ is larger 
than the limit case. In the grey region, the function $\Omega$ is strictly increasing 
for positive $D$ as the maximal value is achieved for some negative $D$.

Figure~\ref{fig:TPvsD} displays the evolution of $\Omega$ as function of $D$ for four different 
choices of the parameters $r$ and $k$, labeled from \textbf{a} to \textbf{d} in Fig.~\ref{fig:btp}. The plots have been obtained by implementing a pseudo arc-length 
continuation to the equation of the equilibrium points starting at the point $(1,k)$ for $D=0$. The program
is terminated when the characteristic curve crosses the homotopy level $\{D=10\}$. 
This program has been also  
used to check the correctness of the formulas obtained in this section. Figure~\ref{fig:TPvsD}(a) has been obtained with $r=1.5$ and $k=1.5$. At $D=0$ the 
value of $\Omega$ is $1+k=2.5$. The characteristic curve increases until some maximal value that 
can be computed using formula~\eqref{eq:OTP} (in this case $\Omega^{*}\approx 2.5247$). After the 
maximal value is reached, the curve decreases monotonically until a horizontal asymptote. This 
asymptote can be computed by using the asymptotic formula of Remark~\ref{rmk:lc}. In this case
$\Omega$ decreases asymptotically to the value $2.5$ (i.e. the sum of carrying capacities: $1+k$).
In Fig.~\ref{fig:TPvsD}(b) we set $r=0.5$ and $k=1.5$. These are conditions 
where the theory predicts that positive positive dispersal rate is detrimental to the total population. 
Indeed, the function $\Omega$, in this case, is strictly decreasing until the asymptotic limit $\Omega = 2.25$. Notice
that $\Omega^{*}$ is achieved for some $D<0$. 
In Fig.~\ref{fig:TPvsD}(c) we have selected $r=1.5$ and $k=k^{*}(1.5)$, where
$k^{*}$ is taken as in the statement of Corollary~\ref{cor:btp}. As it can be seen in the plot, the theory
predicts that the maximal value of $\Omega$ is achieved at the limit case and, therefore, the function 
$\Omega$ is strictly increasing for positive values of $D$. The asymptotic value $\Omega$ is, in this 
case, $2.25$. This value can be predicted with both Theorem~\ref{thm:odr} and Remark~\ref{rmk:lc}.
Finally, Fig.~\ref{fig:TPvsD}(d) illustrates a case in which the limit of $\Omega$ is below the sum of the carrying capacities. The values of the parameters are $r=1.5$
and $k=1.7$. The maximal total population is slightly larger ($\approx 2.72$) and it is achieved 
for some positive $D$.

\section{Conclusions}\label{sec:con}
The investigation of metapopulation mathematical models is very extensive in the literature. Most of these studies have focused on small metapopulations considering few pacthes~\cite{Hastings1993,Arditi2015,Arditi2016, Sardanyes2010}, or in multi-patch systems~\cite{Allen1993}. Discrete-time metapopulation models have focused on inspecting the role of dispersal in the stability of chaotic dynamics and its role in metapopulations' persistence~\cite{Hastings1993,Allen1993}. For instance, Allen and co-workers showed that local chaotic dynamics involved lower extinction probabilities under both intrinsic and global noise~\cite{Allen1993}. Additionally, two-patch time-continuous models have been extensively investigated. Arditi and colleagues used two coupled logistic systems to study the total population under arbitrarily large dispersal rates~\cite{Arditi2015}. Later on they explored the same system of coupled logistic models using the balanced dispersal model instead of linear diffusion~\cite{Arditi2016}. Moreover, the exploration of two-patch models with a generic growth function revealed certain conditions of optimality, indicating that the total population can surpass the sum of carrying capacities of each independent patch~\cite{Holt1985}. More recently, a similar system was explored for local populations growing hyperbolically instead of exponentially~\cite{Sardanyes2010}.

In this manuscript we have revisited one of the simplest metapopulation models, the one in which a diagonal connectivity matrix couples two logistic differential equations. The investigation of metapopulations with few patches allows developing analytical studies, providing clear information on dynamical phenomena such as equilibrum points and bifurcations that could be conserved in higher dimensions. Unlike the previous works on two-patch, time-continuous metapopulations, we have focused on global aspects of the system. Our work provides two novel main contributions to metapopulation theory. 

The first contribution has to do with the robustness of the system to external perturabations. That is, the 
capacity of the system to recover from a situation in which one of the patches is close 
to extinction. It has been known for decades that, for all the values of the parameters, 
an empty patch will be eventually populated by individuals from the other patch. In this process,  
the populations of both patch tend to stabilize at some coexistence equilibrium point. 
However, references in literature do not take into account that the stability type of 
the global extinction equilibrium point (namely, the origin) 
can play a relevant role in the fragility of the metapopulation in terms of extinction. Concerning this point, we have identified two different scenarios: 
(i) when the origin is a source, the trajectories with initial conditions close to the axes (i.e. 
situations in which one of the patches is close to extinction) need some time to separate 
from them; (ii) when the origin is a saddle, there is a unique, locally attracting, 
heteroclinic connection between the origin and the coexistence equilibrium that leads low initial condition to achieve a safe state in a short period of time. By including noise in the dynamics we have shown that recovery is more prone when the origin is a saddle point. Notice that, 
at first sight, the origin being a source could seem a better situation to prevent 
from extinction under perturbations. To illustrate this fact we have run a simulation in 
which the system has stochastic losses of individuals in both patches.  

The second contribution is related to a well known counter-intuitive property of the system: 
under suitable conditions, having positive dispersal rate, leads the total population to 
stabilize overcoming the sum of the carrying capacities of the patches. These phenomenon was 
established for the limit case under the $k-r$ positiveness hypothesis~\cite{Freedman1977,Holt1985, Arditi2015}. That is, if we 
name $r_1$ and $k_1$ the growth rate and the carrying capacity of the first patch, and $r_2$, $k_2$ the respective quantities related to the second patch, $k-r$
positiveness means that $k_1 > k_2$ and $r_1 k_2 > r_2 k_1$ or $k_2 > k_1$ and $r_2 k_1 > r_1 k_2$.
On the other hand, a less restrictive hypothesis was identified in \cite{ruizherrera2018}. If
$k_1 > k_2$ and $r_1 > r_2$ (or $k_1 < k_2$ and $r_1 < r_2$) and $D$ is small enough, the 
the total population exceeds also the sum of the carrying capacities. 

In this work we study analytically the problem for all the positive values of the 
dispersal rate. In particular, we derive a formula for the optimal dispersal rate and 
a bound (that is fulfilled for the optimal dispersal rate) to the total population. To 
avoid cumbersome notation we have used dimensionless units but we stress that the 
change of units can be pulled-back to recover our formulas for standard units. 
The bound to the total population reads as:
$$\overline{\Omega}^{*}=\frac{1}{2} \left ( \sqrt{ (k_{1}+k_{2})^{2} + \frac{(r_{1}-r_{2})^{2}}{r_1 r_2} k_1 k_2} +k_1 + k_2 \right), ~~~\overline{D}^{*}=D^{*}/r_1.$$

The main takeaway Section~\ref{sec:odr} is a unified framework in which both $r$-$k$ positiveness and 
the less restrictive conditions appearing in \cite{Ruiz-Herrera2018} and all 
positive values of the dispersal rate are included. In particular, there are regions in which 
the $r$-$k$ positiveness is fulfilled but the maximal population is achieved for some finite 
value of $D$. We derive formulae that can be used to predict both the optimal dispersal rate, 
the maximal population and the asymptotic behaviour of the function $\Omega$. 

The simplicity of the model studied in this work has allowed us to conduct an analytic description 
of global aspects of the system (namely, global dynamics and behaviour of the total population 
as a function of dispersal). We remark, however, that the phenomena investigated in this paper e.g., that changes in local behavior can impact on global dynamics, may appear in other complex systems. For instance, such results found in low-dimensional metapopulation models may be found in higher dimensions.

\section*{Technical details}
The programs corresponding to Sections \ref{sec:kicks} and \ref{sec:odr} have been written in C from 
the scratch and are available upon request. Library LAPACK~\cite{lapack99} has been used for linear algebra and the Taylor package~\cite{JorbaZ05}) has been used to perform the numerical integrations of system~\eqref{eq:nf}.  

\section*{Declarations}
 
\textbf{Ethical Approval}\\
Not applicable
 \\
 \\
\textbf{Competing interests} 
\\
The authors declare no competing interests.
\\ 
\\
\textbf{Authors' contributions}
\\
MJ, RO and DP carried out the mathematical calculations. MJ, DP, and JS performed the numerical simulations. All authors wrote and reviewed the manuscript. 
 \\
 \\
\textbf{Funding}
\\
This work has been supported by the Spanish State Research Agency (AEI) through grants PGC2018-100928-B-I00 (DP), PGC2018-100699-B-I00 (MCIU/AEI/FEDER, UE) (MJC), and by Catalan Government grant 2017 SGR 1374 (MJC). JS has been supported by the Ram\'on y Cajal grant RYC-2017-22243 funded by MCIN/AEI/10.13039/501100011033 ”FSE invests in your future”, and by the 2020-2021 Biodiversa+ and Water JPI joint call under the BiodivRestore ERA-NET Cofund (GA N°101003777) project MPA4Sustainability with funding organizations: Innovation Fund Denmark (IFD), Agence Nationale de la Recherche (ANR), Fundação para a Ciência e a Tecnologia (FCT), Swedish Environmental Protection Agency (SEPA), and grant PCI2022-132926 funded by MCIN/AEI/10.13039/501100011033 and by the European Union NextGenerationEU/PRTR. This work has been also funded through the Severo Ochoa and Mar\'ia de Maeztu Program for Centers and Units of Excellence in R\&D (CEX2020-001084-M). We thank CERCA Programme/Generalitat de Catalunya for institutional support. DP has been supported by the Research Program 2021-2023 of Universidad Internacional de la Rioja. RO thanks to the Faculty for the Future program of the Schlumberger Foundation for the scholarship.
\\
\\
\textbf{Availability of data and materials}
\\
The software developed for this article is available upon request.

\bibliographystyle{apalike}
\bibliography{cobi}
\end{document}